\documentclass[10pt]{article}
\usepackage{amsmath, amssymb, amsthm, authblk, bigints, fullpage, hyperref, graphicx, color}

\allowdisplaybreaks

\newcommand{\beq}{\begin{equation}}
\newcommand{\eeq}[1]{\label{#1}\end{equation}}

\newtheorem{theorem}{Theorem}[section]
\newtheorem{remark}[theorem]{Remark}
\newtheorem{lemma}[theorem]{Lemma}

\begin{document}
\title{\sc Two-scale series expansions for travelling wave packets in one-dimensional periodic media}
\author[1]{Kirill D. Cherednichenko}
\affil[1]{Department of Mathematical Sciences, University of Bath, Claverton Down, Bath, BA2 7AY, UK}

\maketitle

\begin{abstract} Starting from the wave equation for a continuous medium with material properties that vary 
periodically, we study a system of recurrence relations for a novel series expansion describing the propagation of wave packets oscillating on the microscale ({\it i.e.} on lengths of the order of the period of the medium) and varying slowly on the macroscale ({\it i.e.} on lengths that contain a large number of periods). The resulting equations contain a version of the geometric optics and a description of the overall energy transport through the medium. We illustrate the developed asymptotic theory using the example of a point pulse propagating through a periodic arrangement of two materials with highly contrasting elastic moduli. 
 \end{abstract}
 
 \
 

\section{Introduction}

In what follows we develop an asymptotic framework for the analysis of wave propagation through periodic composites. We derive formal asymptotic expansions for solutions to hyperbolic differential equations describing propagation of oscillatory wave packets through such composites, when the period of the medium (``microscopic" length) is small in comparison with the typical range over which the spatial envelope of the packet varies (``macroscopic" length). 

Denote by $\varepsilon$ the ratio between the above lengths. We are interested in wave packet solutions $u=u(x,t,\varepsilon)$ ({\it i.e.} $u(\cdot, t,\varepsilon)\in L^2({\mathbb R}^d),$ for all $(t,\varepsilon$)) to the wave equation 
\begin{equation}
u_{tt}-\nabla_x\cdot\bigl(A(x, x/\varepsilon, t)\nabla_xu\bigr)=0,\qquad (x, t)\in{\mathbb R}^{d+1},
\label{orig_eqn}
\end{equation}
where $d\in\{1,2,3\},$ $A=A(x, y, t)$ is a uniformly positive-definite ($A>0$) symmetric ($A^\top=A$)
$y$-periodic smooth matrix function. 
Our aim is to develop a general asymptotic theory for wave packets that oscillate on the scale $\varepsilon$ of the period in the coefficients of (\ref{orig_eqn}) with an amplitude that ``varies slowly'', {\it i.e.} whose gradient is of the order $O(1)$ as $\varepsilon\to 0.$ 

Our formal asymptotic expansion (see (\ref{uform}), (\ref{expansion})) is a generalisation of the series adopted by Allaire, Palombaro and Rauch in \cite{APR_2009}, \cite{APR_2011}, \cite{APR_2013}, where 
solutions to the Cauchy problem for (\ref{orig_eqn}) with ``specially prepared'' initial conditions are 
analysed in the time regimes $t=O(1),$ $t=O(\varepsilon^{-1})$ as $\varepsilon\to0.$ We do not make any assumptions about the initial data, except for those that are sufficient for the existence of solutions to the differential equations in question.

The substitution of the expansion (\ref{expansion}) into (\ref{orig_eqn}) results in a sequence of recurrence relations for the phase functions $\phi,$ $\eta$ and the amplitudes ${\mathcal U}^{(n)}$. The nonlinear ``eikonal" equation (\ref{phi_0_eqn}) for the function $\phi$ is similar to \cite[Eq. 35]{APR_2011} but involves in addition the dependence of its solution on the ``quasimomentum''  $\varkappa$ in the integral representation (\ref{uform}).
It turns out that the leading-order amplitude ${\mathcal U}^{(0)}$ has the form (\ref{U0form}), where the variables $x$ and $y$ are separated: the slowly varying factor $u^{(0)}$ satisfies the transport equation 
(\ref{u0_final}), while the rapidly oscillating factor $U^{(0)}$ satisfies the eigenvalue problem 
(\ref{eigen0})--(\ref{eigen0_norm}) on the period cell of the composite. Our analysis of the leading-order term of the expansion (\ref{expansion}) for the``Gelfand transform'' $\widehat{u},$ related to the solution $u$ via the formula (\ref{uform}), is supplemented by the use of the method of stationary phase, which provides the general asymptotic form (see Section \ref{asymptotic_form}) for a rapidly oscillating wave train with a slowly varying amplitude envelope. 

We complement our analysis by the derivation of the propagation properties of the 
quasimomentum $\varkappa$, the ``local wavenumber'' $k$ and the ``energy'' 
\[
\int_{{\mathbb R}^d}\vert u\vert^2,
\]
see Sections \ref{prop_k}, \ref{prop_amp}. We show that these propagate with the ``group velocity'' 
$\Omega'(\varkappa)$ calculated on the basis of the dispersion relation $\Omega=\Omega(\varkappa)$ at a given ``macroscopic location'' $x,$ see (\ref{phi_0_eqn}), where $\Omega$ is the eigenvalue in (\ref{eigen0}).

Finally, in Section \ref{hc_section} we consider a particular case of our general asymptotic formula when the travelling wave has the form of a $\delta$-function pulse and apply it to the setting of a periodic medium with high contrast (or ``large coupling''). As in this case we can take advantage of an explicit asymptotic formulae for the dispersion relations, we demonstrate, on the basis of the general analysis of the present article, that the pulse propagation through such medium is controlled in a quantitatively explicit way, showing the effect of each component of the composite on the wave.

\section{One-dimensional formulation}

We focus on the (1+1)-version ({\it i.e.} $d=1$) of the equation (\ref{orig_eqn}), where $A$ is a scalar, which we will label by $a$: 
\begin{equation}
u_{tt}-\bigl(a(x, x/\varepsilon, t)u_x\bigr)_x=0.
\label{1D_orig_eqn}
\end{equation}
Here $a=a(x, y, t)$ is a positive and bounded scalar function, 1-periodic 
in $y,$  with a bounded inverse. The equation (\ref{1D_orig_eqn}) may describe anti-plane shear waves in an elastic medium that in general is locally periodic (showing an explicit dependence of the coefficient $a$ on the spatial variable $x$), possibly with time-dependent material properties.
 The asymptotic framework presented below can be generalised in a straightforward way to arbitrary ({\it i.e.} non-polarised) displacement fields in a stratified elastic medium.

Note first that $u=u^\varepsilon$ can be written in the form 
\beq
u^\varepsilon(x,t)=\frac{1}{\sqrt{2\pi\varepsilon}}\int_{-\pi}^{\pi}\widehat{u}\biggl(\frac{x}{\varepsilon}, t, \varkappa, \varepsilon\biggr)\exp\Bigl({\rm i}\varkappa\frac{x}{\varepsilon}\Bigr)d\varkappa,
\eeq{uform}
where $\widehat{u}$ is the scaled Gelfand transform (see \cite{Gelfand}, \cite{BLP}, \cite{GrandePreuve}) of the function $u=u(x, t)$ with respect to the spatial variable $x:$ 
\beq
\widehat{u}(y, t, \varkappa,\varepsilon)=\sqrt{\frac{\varepsilon}{2\pi}}\sum_{n\in{\mathbb Z}}u\bigl(\varepsilon(y+n), t\bigr)
\exp\bigl(-{\rm i}(y+n)\varkappa\bigr),\ \ \ \ \varkappa\in [-\pi,\pi)
\eeq{Floquet_transform}
The equation (\ref{1D_orig_eqn}) implies that for each $\varkappa\in[-\pi, \pi)$ the function 
$\widehat{u}(x/\varepsilon, t, \varkappa, \varepsilon),$ which is $\varepsilon$-periodic in $x,$ satisfies the following equation on the ``$\varepsilon$-period cell" $\varepsilon Q=(0, \varepsilon):$
\beq
\varepsilon^2\widehat{u}_{tt}-\biggl(\varepsilon\frac{d}{dx}+{\rm i}\varkappa\biggr)a(x, x/\varepsilon, t)
\biggl(\varepsilon\frac{d}{dx}+{\rm i}\varkappa\biggr)\widehat{u}=0,
\eeq{uhateqn}
which we analyse in what follows. 

\begin{remark} 
For each value of $\varkappa\in[-\pi, \pi),$ both real and imaginary parts of the ``elementary Floquet solution''
\beq
u(x, t, \varkappa, \varepsilon):=\widehat{u}\biggl(\frac{x}{\varepsilon}, t, \varkappa, \varepsilon\biggr)\exp\Bigl({\rm i}\varkappa\frac{x}{\varepsilon}\Bigr)
\eeq{elem_solution}
are solutions to the original equation (\ref{1D_orig_eqn}).
\label{remark_Floquet}
\end{remark}
The function $\widehat{u}(x/\varepsilon, t, \varkappa, \varepsilon)$ is a convenient representation of the solution 
$u^\varepsilon$ in the case when it is time-harmonic, {\it i.e.} when there exists $\omega\in{\mathbb R}$ such that ({\it cf.} (\ref{uform}))
\begin{equation}
u^\varepsilon(x,t)=\frac{1}{\sqrt{2\pi\varepsilon}}\exp({\rm i}\omega t)\int_{-\pi}^{\pi}\widetilde{u}\biggl(\frac{x}{\varepsilon}, \varkappa, \varepsilon, \omega\biggr)\exp\Bigl({\rm i}\varkappa\frac{x}{\varepsilon}\Bigr)d\varkappa,
\label{time_harmonic}
\end{equation}
where $\widetilde{u}(x/\varepsilon, \varkappa, \varepsilon, \omega)$ is the Gelfand transform (see (\ref{Floquet_transform})) of the initial value $u(x, 0, \varkappa, \varepsilon),$ and (\ref{time_harmonic}) simply describes its constant-frequency time evolution. The transform (\ref{Floquet_transform}) results in a decomposition of the spatial part $-(au_x)_x$ of the differential operator in (\ref{1D_orig_eqn}) into a ``direct integral'' of operators on the cell $Q,$ represented by the spatial $\varkappa$-parametrised part of (\ref{uhateqn}) with 
$\varepsilon(d/dx)$ replaced by $d/dy,$ showing that the representation of the solution by $\varkappa$-dependent components is constant in time. The component of the solution in each $\varkappa$-fibre at the frequency $\omega$ can then be analysed, resulting in the dynamic picture for all times $t$. One could argue then that it should be possible to study the behaviour of solutions to (\ref{1D_orig_eqn})  by applying the inverse Fourier transform in time:
\[
u^\varepsilon(x,t)=\frac{1}{\sqrt{2\pi\varepsilon}}\int_{-\infty}^\infty \exp({\rm i}\omega t)\int_{-\pi}^{\pi}\widetilde{u}\biggl(\frac{x}{\varepsilon}, \varkappa, \varepsilon, \omega\biggr)\exp\Bigl({\rm i}\varkappa\frac{x}{\varepsilon}\Bigr)d\varkappa d\omega.
\] 
Changing the order of integration in the last identity (which can be justified under appropriate assumptions on the data in (\ref{1D_orig_eqn})), we obtain
\begin{equation}
u^\varepsilon(x,t)=\frac{1}{\sqrt{2\pi\varepsilon}}\int_{-\pi}^{\pi}\int_{-\infty}^\infty \exp({\rm i}\omega t)\widetilde{u}\biggl(\frac{x}{\varepsilon}, \varkappa, \varepsilon, \omega\biggr) d\omega\exp\Bigl({\rm i}\varkappa\frac{x}{\varepsilon}\Bigr)d\varkappa,
\label{Fourier_rep}
\end{equation}
which is an alternative form of (\ref{uform}): one simply has 
\begin{equation}
\widehat{u}\biggl(\frac{x}{\varepsilon}, t, \varkappa, \varepsilon\biggr)=
\int_{-\infty}^\infty \exp({\rm i}\omega t)\widetilde{u}\biggl(\frac{x}{\varepsilon}, \varkappa, \varepsilon, \omega\biggr) d\omega
\label{link}
\end{equation}

The formula (\ref{Fourier_rep}) demonstrates how amplitudes $\widetilde{u}(x/\varepsilon, \varkappa, \varepsilon, \omega)$ at each frequency contribute to the solution $u^\varepsilon,$ and thus suggests a formula for solutions of (\ref{1D_orig_eqn}) of the wave-packet type: for a (smooth) function $\psi:{\mathbb R}\to{\mathbb R}$ with compact support, one should instead consider  
\begin{equation}
\frac{1}{\sqrt{2\pi\varepsilon}}\int_{-\pi}^{\pi}\int_{-\infty}^\infty \exp({\rm i}\omega t)\widetilde{u}\biggl(\frac{x}{\varepsilon}, \varkappa, \varepsilon, \omega\biggr)\psi(\omega)d\omega\exp\Bigl({\rm i}\varkappa\frac{x}{\varepsilon}\Bigr)d\varkappa,
\label{wave_packet}
\end{equation}
which obviously also solves (\ref{1D_orig_eqn}) and whose Gelfand transform is given by 
\begin{equation}
\int_{-\infty}^\infty \exp({\rm i}\omega t)\widetilde{u}\biggl(\frac{x}{\varepsilon}, \varkappa, \varepsilon, \omega\biggr)\psi(\omega)d\omega.
\label{omega_Gelfand}
\end{equation}
A comparison with (\ref{link}) immediately reveals the challenge one faces when wishing to understand the spatial structure of the wave-packet (\ref{wave_packet}): the time evolution of the initial value 
\[
\int_{-\infty}^\infty\widetilde{u}\biggl(\frac{x}{\varepsilon}, \varkappa, \varepsilon, \omega\biggr)\psi(\omega)d\omega
\]
of (\ref{omega_Gelfand}) for a fixed $\varkappa$ is obscured by the ``mixing'' of the amplitudes 
$\widetilde{u}(x/\varepsilon, \varkappa, \varepsilon, \omega)$ corresponding to different values of frequency $\omega.$ This, in turn, results in the loss of the dispersive picture for the Gelfand components $\widehat{u}(x/\varepsilon, t, \varkappa, \varepsilon)$ or, put simply, their propagation speeds at various times and locations, $t$ and $x.$  However, it appears that the dependence of $\widehat{u}$ on the parameter $\varepsilon$ has not been fully exploited yet, and perhaps the time evolution in (\ref{elem_solution}) can be separated from its spatial structure in an asymptotically explicit way for small values of $\varepsilon.$
This motivates us to look for new asymptotic ways  
to represent the scale-time separation in (\ref{elem_solution}). 


In our search of a suitable asymptotic tool that would address the above problem, we have been motived by two existing methods.
On the one hand, it is known that in the case of dispersive homogeneous media \cite{Whitham}, when the setup of the operator $-(au_x)_x,$ where $a$ is constant, is generalised to higher-order homogeneous expressions $a\partial_x^{2n},$ $n\in{\mathbb N},$
the analysis at fixed frequencies is ``lifted'' to the global dynamics by applying the Fourier transform in $x$ and utilising an inverse transform with respect to the wavenumber $k$ related to the frequency $\omega$ via $\omega=\pm ak^n,$ resulting in an asymptotic description of the wave dispersion for large times. 

On the other hand, the presence of a small parameter $\varepsilon$ corresponding to the ratio of lengths (that is, of the period of the medium and a macroscopic scale) brings up an analogy with the WKB expansion in high-frequency wave propagation \cite{Babich_Buldyrev, Babich_Buldyrev_Molotkov}, which can be viewed as a generalisation of the plane-wave ansatz to the case of inhomogeneous media.

The above considerations lead us to an alternative representation for (\ref{elem_solution}), which elucidates the dependence of the function $\widehat{u}$ on the parameter $\varepsilon,$ albeit at the price of considering two-scale power series in $\varepsilon.$  
We next introduce such a representation. 



\section{Multiscale version of the WKB asymptotic expansion}
\label{time_independent}


We are looking to determine a more specific form of the solution $\widehat{u}$ to (\ref{uhateqn}), as a (formal) asymptotic expansion in powers of the small parameter $\varepsilon$:
\begin{equation}
\widehat{u}(x, {t},\varkappa,\varepsilon)=\exp\biggl(-\frac{\rm i}{\varepsilon}{\phi}(x, {t}, \varkappa)-{\rm i}{\eta}(x, {t}, \varkappa)\biggr)
\sum_{n=0}^\infty\varepsilon^n{\mathcal U}^{(n)}\Bigl(x,\frac{x}{\varepsilon}, {t}, \varkappa\Bigr),
\label{expansion}
\end{equation}
where the phase functions $\phi,$ $\eta$ are  
real-valued and the amplitude terms ${\mathcal U}^{(n)}(x, y, {t},\varkappa),$ $n=0,1,\dots$ are $Q$-periodic with respect to $y.$

\begin{remark}  
\label{nondeg_remark}
Consider the phase function in (\ref{expansion}):
\[
\Phi(x, {t}, \varkappa, \varepsilon):=\varepsilon^{-1}{\phi}(x, {t}, \varkappa)+{\eta}(x, {t}, \varkappa).
\]
While writing $\widehat{u}$ in the form (\ref{expansion}), we make an implicit ``slow time'' assumption that the phase velocity $\Phi_x^{-1}\Phi_{t}$ is much smaller than
$\varepsilon^{-1}$ as $\varepsilon\to 0.$  In other words, the ``typical relaxation time'' of the solution $u$ is much larger than 
the typical time it takes for the solution to travel across the $\varepsilon$-periodicity cell. 
\end{remark}

In the following we always assume that
the derivatives $\phi_x,$ $\eta_x,$ which describe the
``local wavenumber'' of the solution (\ref{expansion}), are uniformly bounded below in absolute value for $(x, t , \varkappa)
\in{\mathbb R}\times{\mathbb R}\times [-\pi, \pi).$ Under this assumption the requirement of Remark \ref{nondeg_remark} is satisfied, and in fact for small values of $\varepsilon$ the phase velocity 
$\Phi_x^{-1}\Phi_{t}$ is of the order $Q(1).$
 
   To simplify the presentation of the asymptotic analysis, we shall assume that the coefficient $a=a(x, y, t)$ is independent of the first, ``slow", variable $x.$ The discussion of the general case of  $x$-dependence (representing a ``locally periodic" medium)
   carries through in a similar way, subject to minor modifications. As we shall see in what follows, in the particular case when the coefficient $a$ is independent of $x,$ the leading-order phase  ${\phi}$ in (\ref{expansion}) 
 is also independent of $x.$ Furthermore,  when the coefficient $a$ is independent of both 
 $x$ and $t,$ the leading-order amplitude
 ${\mathcal U}^{(0)}$ in (\ref{expansion}) has the form of a travelling wave propagating with velocity $\Omega'(\varkappa),$ and therefore (\ref{expansion}) describes the asymptotic behaviour of the unitary semigroup for (\ref{1D_orig_eqn}) by representing its evolution in the ``fibres'' parametrised by $\varkappa\in[-\pi, \pi).$

\subsection{Eikonal equation for the phase function}

Substituting the expansion (\ref{expansion})
into the equation (\ref{uhateqn}) yields a
system of recurrence relations for the amplitudes ${\mathcal U}^{(n)}={\mathcal U}^{(n)}(\varepsilon y, y, {t}, \varkappa),$ $n=0,1,\dots,$ 
and the phase 
coefficients $\phi=\phi(\varepsilon y, {t}, \varkappa),$ $\eta=\eta(\varepsilon y, {t}, \varkappa),$
 in the series (\ref{expansion}).
Here $\varepsilon y=x$ and $y=x/\varepsilon$ are the ``slow'' and ``fast'' variables, respectively, in the spatial behaviour of the solution $u$ to the equation (\ref{1D_orig_eqn}).

In particular, at the order $\varepsilon^0$ we obtain 
\beq
-(\phi_{t})^2{\mathcal U}^{(0)}(x, y, t, \varkappa)-\biggl(\frac{d}{dy}+{\rm i}\varkappa-{\rm i}{\phi}_x\biggr)
a(y, t)\biggl(\frac{d}{dy}+{\rm i}\varkappa-{\rm i}{\phi}_x\biggr){\mathcal U}^{(0)}(x, y, t, \varkappa)=0,
\eeq{eigen00}
{\it i.e.} on the microscale, in the vicinity of the point $x$ at time ${t},$ the leading-order amplitude ${\mathcal U}^{(0)}$ behaves 
as a $Q$-periodic eigenfunction of the differential operator 
\[
-\biggl(\frac{d}{dy}+{\rm i}\varkappa-{\rm i}{\phi}_x\biggr)a(y,t)\biggl(\frac{d}{dy}+{\rm i}\varkappa-{\rm i}{\phi}_x\biggr),
\] 
corresponding to the value $\varkappa$ of the ``quasimomentum''. 


Hence we write
\beq
{\mathcal U}^{(0)}(x, y, {t}, \varkappa)=u^{(0)}(x, {t}, \varkappa)
U^{(0)}\bigl(y, \Omega(t, \varkappa-{\phi}_x)\bigr),
\eeq{U0form}
and the function ${\phi}$ solves the equation 
\beq
\phi_{t}=\pm\Omega\bigl(t, \varkappa-{\phi}_x\bigr).
\eeq{phi_0_eqn}
where $U^{(0)}=
U^{(0)}\bigl(y, \Omega)$ is a normalised $Q$-periodic eigenfunction in (\ref{eigen00}) corresponding to the eigenvalue $\Omega^2=\Omega^2\bigl(t, \varkappa-{\phi}_x\bigr),$ {\it i.e.} one has
\begin{align}
-\biggl(\frac{d}{dy}+{\rm i}\xi\biggr)
a(y,t)\biggl(\frac{d}{dy}+{\rm i}\xi\biggr)U^{(0)}(y)=\Omega^2(t,\xi)U^{(0)}(y),\label{eigen0}\\[0.2em]
\int_0^1\bigl\vert U^{(0)}(y)\bigr\vert^2dy=1,\label{eigen0_norm}
\end{align}
where $\xi=\varkappa-\phi_x$ plays the role of a parameter. Since for fixed $\xi\in[-\pi, \pi)$ and $\Omega^2(t,\xi)$ the problem (\ref{eigen0})--(\ref{eigen0_norm}) has at most one solution, the eigenfunctions $U^{(0)}$ can be parametrised by $\Omega(t,\xi)$ (at least locally in $\xi$), which is reflected in the notation for $U^{(0)}$ so it depends on $t, x, \varkappa$ via the function $\Omega=\Omega(t, \varkappa-{\phi}_x(t, x,\varkappa))$ only.

In what follows, we choose $\Omega$ to be the positive square root of the eigenvalue $\Omega^2$ in 
(\ref{eigen0}). In formulae containing ``$\pm$" or ``$\mp$", the upper sign corresponds to the choice of ``$+$" in (\ref{phi_0_eqn}) and the lower sign corresponds to the choice of ``$-$''.

Differentiating (\ref{phi_0_eqn}) with respect to ${t}$ and $x$ in turn, we obtain
\begin{equation}
\phi_{tt}=\pm\Omega_t(t, \varkappa-{\phi}_x)\mp{\phi}_{x{t}}\Omega_\xi(t, \varkappa-{\phi}_x)
\label{eq1}
\end{equation}
and 
\begin{equation}
\phi_{tx}=\pm\bigl[\Omega(t, \varkappa-{\phi}_x)\bigr]_x=:\pm\Omega_x,
\label{eq2}
\end{equation}
respectively. Henceforth, the values of $\Omega=\Omega(t,\xi)$ and its derivatives are always  taken 
at the point \linebreak $\varkappa-{\phi}_x(x, {t}, \varkappa).$  Combining (\ref{eq1}) and (\ref{eq2})  yields
\beq
\phi_{tt}=\pm\Omega_t-\Omega_\xi\Omega_x.
\eeq{aux}
We will use this result to simplify the equation (\ref{amplitude_transformed}) in the next section.

\begin{remark}
Using the fact that 
\begin{equation}
\Omega_x=-\Omega_\xi{\phi}_{xx},
\label{second_fact}
\end{equation}
 we obtain a quasilinear hyperbolic equations on ${\phi},$ with local wave speed $|\Omega_\xi|:$ 
\begin{equation}
\phi_{tt}=(\Omega_\xi)^2{\phi}_{xx}\pm\Omega_t.
\label{hyperbolic}
\end{equation}
\end{remark}




\subsection{Transport equation for the amplitude envelope}

Continuing the procedure described at the beginning of the previous section, we collect the terms of order $\varepsilon^1$, which yields 
\begin{align*}
&-\phi_{t}^2\,{\mathcal U}^{(1)}-\biggl(\frac{d}{dy}+{\rm i}\varkappa-{\rm i}{\phi}_x\biggr)
a(y,t)\biggl(\frac{d}{dy}+{\rm i}\varkappa-{\rm i}{\phi}_x\biggr){\mathcal U}^{(1)}\\[0.4em]
&=({\rm i}{\eta}_{{t}{t}}+2\phi_{t}\eta_{t})\,{\mathcal U}^{(0)}+2{\rm i}
\phi_{t}\,{\mathcal U}^{(0)}_{t}-{\rm i}{\phi}_{xx}a(y,t){\mathcal U}^{(0)}
\\[0.3cm]
&
+\biggl\{a(y,t)\biggl(\frac{d}{dy}+{\rm i}\varkappa-{\rm i}{\phi}_x\biggr)+
\biggl(\frac{d}{dy}+{\rm i}\varkappa-{\rm i}{\phi}_x\biggr)a(y,t)\biggr\}{\mathcal U}^{(0)}_x\\[0.2cm]
&-{\rm i}{\eta}_x\biggl\{a(y,t)\biggl(\frac{d}{dy}+{\rm i}\varkappa-{\rm i}{\phi}_x\biggr)+
\biggl(\frac{d}{dy}+{\rm i}\varkappa-{\rm i}{\phi}_x\biggr)a(y,t)\biggr\}{\mathcal U}^{(0)}.
\end{align*}
We re-write this equation using the representation (\ref{U0form}), as follows:
\begin{equation}
\begin{aligned}
-\phi_{t}^2\,{\mathcal U}^{(1)}&-\biggl(\frac{d}{dy}+{\rm i}\varkappa-{\rm i}{\phi}_x\biggr)
a(y,t)\biggl(\frac{d}{dy}+{\rm i}\varkappa-{\rm i}{\phi}_x\biggr){\mathcal U}^{(1)}\\[0.3em]
&=({\rm i}{\eta}_{{t}{t}}+2\phi_{t}\eta_{t})u^{(0)}U^{(0)}+2{\rm i}
\phi_{t}u^{(0)}_{t} U^{(0)}-2{\rm i}\phi_{t}u^{(0)}U^{(0)}_\Omega\Omega_\xi
{\phi}_{x{t}}\\[0.4em]
&-{\rm i}{\phi}_{xx}u^{(0)}a(y,t)U^{(0)}+\biggl\{a(y,t)\biggl(\frac{d}{dy}+{\rm i}\varkappa-{\rm i}{\phi}_x\biggr)+
\biggl(\frac{d}{dy}+{\rm i}\varkappa-{\rm i}{\phi}_x\biggr)a(y,t)\biggr\}u^{(0)}_xU^{(0)}\\[0.3em]
&+\biggl\{a(y,t)\biggl(\frac{d}{dy}+{\rm i}\varkappa-{\rm i}{\phi}_x\biggr)+
\biggl(\frac{d}{dy}+{\rm i}\varkappa-{\rm i}{\phi}_x\biggr)a(y,t)\biggr\}u^{(0)}U^{(0)}_\Omega{\Omega_\xi}
{\phi}_{xx}\\[0.3em]
&-{\rm i}{\eta}_x\biggl\{a(y,t)\biggl(\frac{d}{dy}+{\rm i}\varkappa-{\rm i}{\phi}_x\biggr)+
\biggl(\frac{d}{dy}+{\rm i}\varkappa-{\rm i}{\phi}_x\biggr)a(y,t)\biggr\}u^{(0)}U^{(0)},
\end{aligned}
\label{U1eqn}
\end{equation}
where $U^{(0)}_\Omega$ stands for the derivative of $U^{(0)}(y, \Omega)$ with respect to $\Omega.$ 

We treat (\ref{U1eqn}) as an equation for ${\mathcal U}^{(1)}$ and so seek the condition of solvability for it.
To this end, we multiply (\ref{U1eqn}) by the complex conjugate of the function $U^{(0)}$ found at the previous step and integrate the result with respect to the variable $y\in Q=(0,1).$  Using the eigenfunction equation 
(\ref{eigen0}) 
yields
\begin{equation}
\begin{aligned}
({\rm i}{\phi}_{{t}{t}}&+2\phi_{t}\eta_{t})u^{(0)}+
2{\rm i}\phi_{t}u^{(0)}_{t}-2{\rm i}\phi_{t}u^{(0)}{\Omega_\xi}
{\phi}_{x{t}}\int_0^1U^{(0)}_\Omega\overline{U^{(0)}}-{\rm i}{\phi}_{xx}\int_0^1a(y,t)\bigl\vert U^{(0)}\bigr\vert^2\\[0.3em]
&+u^{(0)}_x\int_0^1\biggl\{a(y,t)\biggl(\frac{d}{dy}+{\rm i}\varkappa-{\rm i}{\phi}_x\biggr)+
\biggl(\frac{d}{dy}+{\rm i}\varkappa-{\rm i}{\phi}_x\biggr)a(y,t)\biggr\}U^{(0)}\overline{U^{(0)}}\\[0.3em]
&-u^{(0)}{\Omega_\xi}{\phi}_{xx}\int_0^1\biggl\{a(y,t)\biggl(\frac{d}{dy}+{\rm i}\varkappa-{\rm i}{\phi}_x\biggr)+
\biggl(\frac{d}{dy}+{\rm i}\varkappa-{\rm i}{\phi}_x\biggr)a(y,t)\biggr\}U^{(0)}_\Omega\overline{U^{0}}\\[0.3em]
&-{\rm i}{\eta}_xu^{(0)}\int_0^1\biggl\{a(y,t)\biggl(\frac{d}{dy}+{\rm i}\varkappa-{\rm i}{\phi}_x\biggr)+
\biggl(\frac{d}{dy}+{\rm i}\varkappa-{\rm i}{\phi}_x\biggr)a(y,t)\biggr\}U^{(0)}\overline{U^{(0)}}=0.
\end{aligned}
\label{solv_cond}
\end{equation}

We would like to separate the real and imaginary parts of (\ref{solv_cond}). To this end, the following observation proves useful.
\begin{lemma}
The expression
\[
-{\rm i}\int_0^1U^{(0)}_\Omega\overline{U^{(0)}}
\]
is real-valued.
\end{lemma}
\begin{proof}
Notice that
\begin{align*}
0&=\frac{d}{d\Omega}\int_0^1U^{(0)}\overline{U^{(0)}}=\int_0^1U^{(0)}_\Omega\overline{U^{(0)}}+\int_0^1U^{(0)}\overline{U^{(0)}_\Omega}
\\[0.3em]
&=\int_0^1U^{(0)}_\Omega\overline{U^{(0)}}+\overline{\int_0^1\overline{U^{(0)}}U^{(0)}_\Omega}=2\Re\int_0^1U^{(0)}_\Omega\overline{U^{(0)}},
\end{align*}
which immediately yields the claim.
\end{proof}

Considering the real part of (\ref{solv_cond}) results in an equation for $\eta$
as follows:
\begin{equation}
\begin{aligned}
2\phi_{t}\eta_{t}&+2\phi_{t}{\Omega_\xi}{\phi}_{x{t}}\biggl(-{\rm i}\int_0^1U^{(0)}_\Omega\overline{U^{(0)}}\biggr)
\\[0.3em]
&-{\Omega_\xi}{\phi}_{xx}\Re\int_0^1\biggl\{a(y,t)\biggl(\frac{d}{dy}+{\rm i}\varkappa-{\rm i}{\phi}_x\biggr)+
\biggl(\frac{d}{dy}+{\rm i}\varkappa-{\rm i}{\phi}_x\biggr)a(y,t)\biggr\}U^{(0)}_\Omega\overline{U^{0}}\\[0.3em]
&+2{\eta}_x\Im\int_0^1a(y,t)\biggl(\frac{d}{dy}+{\rm i}\varkappa-{\rm i}{\phi}_x\biggr)U^{(0)}\overline{U^{(0)}}=0,
\end{aligned}
\label{real_part}
\end{equation}
while taking the imaginary part of (\ref{solv_cond}) yields the following equation for $u^{(0)}:$
\begin{equation}
\begin{aligned}
{\phi}_{{t}{t}}u^{(0)}&+2\phi_{t}u^{(0)}_{t}-{\phi}_{xx}\int_0^1a(y,t)\bigl\vert U^{(0)}\bigr\vert^2
+2u^{(0)}_x\Im\int_0^1a(y,t)\biggl(\frac{d}{dy}+{\rm i}\varkappa-{\rm i}{\phi}_x\biggr)U^{(0)}\overline{U^{(0)}}\\[0.3em]
&-{\Omega_\xi}{\phi}_{xx}u^{(0)}\Im\int_0^1\biggl\{a(y,t)\biggl(\frac{d}{dy}+{\rm i}\varkappa-{\rm i}{\phi}_x\biggr)+
\biggl(\frac{d}{dy}+{\rm i}\varkappa-{\rm i}{\phi}_x\biggr)a(y,t)\biggr\}U^{(0)}_\Omega\overline{U^{0}}.
\end{aligned}
\label{imag_part}
\end{equation}
Notice that 
\begin{align*}
\biggl(\Im\int_0^1a(y,t)\biggl(\frac{d}{dy}+{\rm i}\varkappa-{\rm i}{\phi}_x\biggr)U^{(0)}\overline{U^{(0)}}\biggr)_x=&
-{\phi}_{xx}\int_0^1a(y,t)\bigl\vert U^{(0)}\bigr\vert^2\\[0.3em]
&-{\Omega_\xi}{\phi}_{xx}\Im\int_0^1a(y,t)\biggl(\frac{d}{dy}+{\rm i}\varkappa-{\rm i}{\phi}_x\biggr)U^{(0)}_\Omega\overline{U^{(0)}}\\[0.3em]
&-{\Omega_\xi}{\phi}_{xx}\Im\int_0^1a(y,t)\biggl(\frac{d}{dy}+{\rm i}\varkappa-{\rm i}{\phi}_x\biggr)U^{(0)}\overline{U^{(0)}_\Omega},
\end{align*}
which is the sum of the third and fifth terms in (\ref{imag_part}).
Hence, we obtain
\begin{align*}
{\phi}_{{t}{t}}u^{(0)}+2\phi_{t}u^{(0)}_{t}
&+2u^{(0)}_x\Im\int_0^1a(y,t)\biggl(\frac{d}{dy}+{\rm i}\varkappa-{\rm i}{\phi}_x\biggr)U^{(0)}\overline{U^{(0)}}\\[0.3em]
&+u^{(0)}\biggl(\Im\int_0^1a(y,t)\biggl(\frac{d}{dy}+{\rm i}\varkappa-{\rm i}{\phi}_x\biggr)U^{(0)}\overline{U^{(0)}}\biggr)_x=0,
\end{align*}
or, after multiplication by $u^{(0)}$ and using the product rule,
\beq
\Bigl[\bigl(u^{(0)}\bigr)^2{\phi}_{t}\Bigr]_{t}+\biggl[\bigl(u^{(0)}\bigr)^2\Im\int_0^1a(y,t)\biggl(\frac{d}{dy}+{\rm i}\varkappa-{\rm i}{\phi}_x\biggr)U^{(0)}\overline{U^{(0)}}\biggr]_x=0,
\eeq{amplitude0}
which is a transport equation for $\bigl(u^{(0)}\bigr)^2.$
Finally, we use the following statement.
\begin{lemma}
\label{lemma_coeff}
The formula
\beq
\Im\int_0^1a(y,t)\biggl(\frac{d}{dy}+{\rm i}\varkappa-{\rm i}{\phi}_x\biggr)U^{(0)}\overline{U^{(0)}}=\frac{1}{2}\bigl(\Omega^2\bigr)_\xi,
\eeq{lemma_coeff_form}
holds, where the right-hand side is evaluated at $\varkappa-{\phi}_x(x,{t},\varkappa).$
\end{lemma}
\begin{proof}
Differentiating with respect to $\varkappa$ the eigenvalue equation (\ref{eigen0}), we obtain
\begin{align*}
&-{\rm i}(1-{\phi}_{x\varkappa})\biggl\{a(y,t)\biggl(\frac{d}{dy}+{\rm i}\varkappa-{\rm i}{\phi}_x\biggr)
+\biggl(\frac{d}{dy}+{\rm i}\varkappa-{\rm i}{\phi}_x\biggr)
a(y,t)\biggr\}U^{(0)}\\[0.3em]
&-\biggl(\frac{d}{dy}+{\rm i}\varkappa-{\rm i}{\phi}_x\biggr)
a(y,t)\biggl(\frac{d}{dy}+{\rm i}\varkappa-{\rm i}{\phi}_x\biggr)\frac{d}{d\varkappa}U^{(0)}=
(1-{\phi}_{x\varkappa})\bigl(\Omega^2\bigr)_\xi\,U^{(0)}+\Omega^2\frac{d}{d\varkappa}U^{(0)}.
\end{align*}
Multiplying both sides of the last equation by $\overline{U^{(0)}},$ integrating by parts in the last term on 
the left-hand side and using once again the eigenvalue equation (\ref{eigen0}) yields
\[
-{\rm i}(1-{\phi}_{x\varkappa})\int_0^1\biggl\{a(y,t)\biggl(\frac{d}{dy}+{\rm i}\varkappa-{\rm i}{\phi}_x\biggr)
+\biggl(\frac{d}{dy}+{\rm i}\varkappa-{\rm i}{\phi}_x\biggr)
a(y,t)\biggr\}U^{(0)}\overline{U^{(0)}}=(1-{\phi}_{x\varkappa})\bigl(\Omega^2\bigr)_\xi.
\]
Finally, we obtain (\ref{lemma_coeff_form}) by noticing that
\begin{align*}
\int_0^1\biggl\{a(y,t)\biggl(\frac{d}{dy}+{\rm i}\varkappa-{\rm i}{\phi}_x\biggr)
&+\biggl(\frac{d}{dy}+{\rm i}\varkappa-{\rm i}{\phi}_x\biggr)
a(y,t)\biggr\}U^{(0)}\overline{U^{(0)}}
\\[0.3em]
&=2{\rm i}\Im\int_0^1a(y,t)\biggl(\frac{d}{dy}+{\rm i}\varkappa-{\rm i}{\phi}_x\biggr)U^{(0)}\overline{U^{(0)}}.
\end{align*}
\end{proof}

Combining (\ref{amplitude0}) and Lemma \ref{lemma_coeff} yields
\beq
\Bigl[\bigl(u^{(0)}\bigr)^2{\phi}_{t}\Bigr]_{t}+\frac{1}{2}\Bigl[\bigl(u^{(0)}\bigr)^2\bigl(\Omega^2\bigr)_\xi\Bigr]_x=0.
\eeq{amplitude}
Using the product rule we re-write (\ref{amplitude}) as 
\beq
({\phi}_{{t}{t}}+{\Omega_\xi}\Omega_x)\bigl(u^{(0)}\bigr)^2+\Bigl[\bigl(u^{(0)}\bigr)^2\Bigr]_{t}{\phi}_{t}+\Bigl[\pm\bigl(u^{(0)}\bigr)^2{\Omega_\xi}\Bigr]_x(\pm\Omega)=0,
\eeq{amplitude_transformed}
where the first term equals $\pm\Omega_t$ in view of (\ref{aux}), and hence
\[
\Bigl[\bigl(u^{(0)}\bigr)^2\Bigr]_{t}{\phi}_{t}+\Bigl[\pm\bigl(u^{(0)}\bigr)^2{\Omega_\xi}\Bigr]_x(\pm\Omega)=\mp\bigl(u^{(0)}\bigr)^2\Omega_t.
\] 
Finally, using (\ref{phi_0_eqn}) results in 
\beq
\Bigl[\bigl(u^{(0)}\bigr)^2\Bigr]_{t}+\Bigl[\pm\bigl(u^{(0)}\bigr)^2{\Omega_\xi}\Bigr]_x=-\bigl(u^{(0)}\bigr)^2(\log\Omega)_t,
\eeq{u0_final}
which is the transport equation for the modulating function $u^{(0)}=u^{(0)}(x, {t}, \varkappa).$ The equation (\ref{u0_final}) is analogous to the amplitude transport equation derived in the theory of linear dispersive waves, {\it cf. e.g.} equation (11.64) in \cite{Whitham}.
 
\subsection{Solution along characteristics}
 
We integrate (\ref{u0_final}) along characteristics parametrised by ${t},$ so that for all $\varkappa\in[-\pi, \pi)$
\begin{align}
&\frac{dx({t})}{d{t}}=\pm{\Omega_\xi},\label{trajectories}\\[0.2cm] 
&\frac{d}{d{t}}\Bigl[\bigl(u^{(0)}(x({t}), {t}, \varkappa)\bigr)^2\Bigr]=\bigl(u^{(0)}(x({t}), {t}, \varkappa)\bigr)^2\Psi(t, \varkappa),
\nonumber
\end{align}
where 
\begin{equation}
\Psi(t,\varkappa):=\Bigl(\mp\Omega_{\xi\xi}(t, \xi){\phi}_{xx}\bigl(x({t}), {t},\varkappa)\bigr)
-\bigl(\log\Omega(t, \xi)\bigr)_t\Bigr)\Bigr\vert_{\xi=\varkappa-{\phi}_{x}(x({t}), {t},\varkappa))}.
\label{Psi}
\end{equation}
It follows that 
\begin{equation}
\bigl(u^{(0)}(x({t}), {t}, \varkappa))\bigr)^2=\bigl(u^{(0)}(x(0),0,\varkappa))\bigr)^2\exp\biggl(\int_0^{t}\Psi(s,\varkappa)ds\biggr).
\label{amplitude_formula}
\end{equation}
Using the hyperbolic equation (\ref{hyperbolic}) for ${\phi},$ we rewrite (\ref{amplitude_formula}) as follows:
\begin{equation}
\begin{aligned}
\bigl(u^{(0)}(x({t}), {t}, \varkappa))\bigr)^2
=\bigl(u^{(0)}(x(0),0,\varkappa))\bigr)^2\exp\biggl(\int_0^{t}\Bigl(&\mp\Omega_{\xi\xi}(\Omega_\xi)^{-2}{\phi}_{ss}\bigl(x(s), s,\varkappa\bigr)\\[0.4em]
&+\Omega_{\xi\xi}(\Omega_\xi)^{-2}\Omega_s-\bigl(\log\Omega(s, \xi)\bigr)_s\Bigr)ds\biggr).
\end{aligned}
\label{ampl_preform}
\end{equation}
where the expression under the integral is evaluated at $\xi=\varkappa-{\phi}_{x}(x({s}), {s},\varkappa),$ {\it cf.} (\ref{Psi}).

For brevity, below 
 we often omit the arguments $x({t}),$ ${t},$ $\varkappa$ of the function ${\phi}$ and its derivatives, as well as the arguments $t,$ $\xi=\varkappa-{\phi}_{x}(x({t}), {t},\varkappa)$ of the functions $\Omega,$ $\Omega_\xi.$ 

\begin{lemma}
\label{lemma_constant}
Suppose that $a=a(t,y).$ Along the characteristics\footnote{These are one-dimensional ``paths" parametrised by ${t}.$} (\ref{trajectories}):

1) The function ${\phi}_x(x({t}), {t}, \varkappa)$ is constant;

2) The following identity holds:
\begin{equation*}
\frac{d}{d{t}}{\phi}\bigl(x({t}), {t}, \varkappa\bigr)=\pm\phi_x\Omega_{\xi}\pm\Omega.
\end{equation*}
\end{lemma}
\begin{proof}
The equations (\ref{eq2}), (\ref{second_fact}) imply
\begin{equation*}
\frac{d}{d{t}}{\phi}_x(x({t}), {t}, \varkappa)=-{\phi}_{xx}x'({t})-{\phi}_{x{t}}=-{\phi}_{xx}(\pm\Omega_\xi)\pm\Omega_\xi{\phi}_{xx}=0,
\end{equation*}
hence the first claim. Furthermore, using the chain rule and identities, we obtain 
\begin{equation*}
\frac{d}{d{t}}{\phi}\bigl(x({t}), {t}, \varkappa\bigr)={\phi}_{x}x'({t})+{\phi}_{{t}},
\end{equation*}
from which the second claim follows using (\ref{phi_0_eqn}), (\ref{trajectories}).
\end{proof}

For each $\varkappa\in [-\pi, \pi),$ denote $g(\sigma):={\phi}(0,\sigma, \varkappa),$ $\sigma\in{\mathbb R},$ the initial values of the leading order phase function ${\phi}.$
The first part of Lemma \ref{lemma_constant} implies that for each $\sigma$ the value $g'(\sigma)$ is ``propagated'' along the characteristics (\ref{trajectories}) as the (constant) value of the derivative ${\phi}_x.$
From the second part of Lemma \ref{lemma_constant}, we infer then that 
\begin{equation}
{\phi}(t,x,\varkappa)=g(\sigma)\pm\int_0^t\Bigl(g'(\sigma)\Omega_\xi\bigl(s,\varkappa-g'(\sigma)\bigr)+\Omega\bigl(s,\varkappa-g'(\sigma)\bigr)\Bigr)ds,
\label{varphi_form}
\end{equation}
where $\sigma$ is related to $x,$ $t$ via 
\begin{equation}
x=\sigma\pm\int_0^t\Omega_\xi\bigl(\tau, \varkappa-g'(\sigma)\bigr)d\tau.
\label{sigmaxt}
\end{equation}

It follows from the above equations that 
\[
{\phi}_t(t, x,\varkappa)=\pm\Omega\bigl(t,\varkappa-g'(\sigma)\bigr),
\]
and hence 
\[
{\phi}_{tt}(t, x,\varkappa)=\pm\Omega_t\bigl(t,\varkappa-g'(\sigma)\bigr)\mp\Omega_\xi\bigl(t,\varkappa-g'(\sigma)\bigr)g''(\sigma),
\]
where $\sigma=\sigma(x,t)$ is given by (\ref{sigmaxt}). Substituting this into (\ref{ampl_preform}) yields
\begin{equation}
\bigl(u^{(0)}(x, {t}, \varkappa)\bigr)^2
=\bigl(u^{(0)}(\sigma,0,\varkappa)\bigr)^2\exp\biggl(-\int_0^{t}\Bigl((\Omega_{\xi\xi}/\Omega_\xi)^2g''(\sigma)+(\log\Omega)_s\Bigr)ds\biggr),
\label{u0_formula}
\end{equation}
where $\Omega=\Omega(s,\xi)=\Omega(s, \varkappa-g'(\sigma)).$

\subsection{The leading-order term of the asymptotics}

\label{asymptotic_form}

Suppose that $g(\sigma)=0,$ $\sigma\in{\mathbb R},$ {\it i.e.} the initial phase vanishes. This choice corresponds to the ``specially prepared'' initial data, whose Gelfand transform has the form 
\begin{equation}
\widetilde{u}(x, \varkappa)U^{(0)}\biggl(\frac{x}{\varepsilon}, \Omega(0, \varkappa)\biggr),\qquad \varkappa\in[-\pi, \pi).
\label{init_spec}
\end{equation}
Then (\ref{varphi_form}), 
(\ref{u0_formula}) read
\begin{align*}
{\phi}(t,x,\varkappa)&=\pm\int_0^t\Omega(s,\varkappa)ds,\\[0.2em]
u^{(0)}(x, {t}, \varkappa)
&=\widetilde{u}\biggl(x\mp\int_0^t\Omega_\varkappa(s,\varkappa)ds, \varkappa\biggr)\sqrt{\frac{\Omega(0,\varkappa)}{\Omega(t,\varkappa)}},
\end{align*}
for some distribution $\widetilde{u}:{\mathbb R}\times [-\pi, \pi)\to{\mathbb R},$ which we assume to be smooth in the first variable.\footnote{The analysis of the case of non-smooth, {\it e.g.} piecewise smooth, $\widetilde{u}$ is outside the scope of this paper.} Furthermore, from (\ref{real_part}) and the form of the initial data (\ref{init_spec}) we have $\eta=0.$

In particular, when $a=a(y),$ and so $\Omega(t,\varkappa)$ is independent of $t,$ we obtain 
\begin{equation}
u^{(0)}(x, {t}, \varkappa)
=\widetilde{u}\bigl(x\mp\Omega'(\varkappa)t,\varkappa\bigr).
\label{t_independent}
\end{equation}

Summarising, the leading-order term in (\ref{uform}), (\ref{expansion}) is given by
\begin{equation}
\begin{aligned}
\int_{-\pi}^{\pi}\widetilde{u}\biggl(x\mp\int_0^t\Omega_\varkappa(s,\varkappa)ds,\varkappa\biggr)&\sqrt{\frac{\Omega(0,\varkappa)}{\Omega(t,\varkappa)}}
U^{(0)}\biggl(\frac{x}{\varepsilon},\Omega(t,\varkappa)\biggr)\\[0.3em]
&\times\exp\biggl[\frac{\rm i}{\varepsilon}\Bigl(\varkappa x\mp\int_0^t\Omega(s, \varkappa)ds\Bigr)\biggr]d\varkappa,
\end{aligned}
\label{elem_asymp}
\end{equation}



The main contribution to the integral (\ref{elem_asymp}) is provided by the neighbourhoods of the points 
$\varkappa=\widehat{\varkappa}$ for which the phase function 
is ``stationary'', {\it i.e.}
\begin{equation}
\pm\int_0^t\Omega_\varkappa(s, \varkappa)ds=x,
\label{stat_eqn1}
\end{equation}
if such points exist, and (\ref{elem_asymp}) is asymptotically smaller than any power of $\varepsilon$ if there are not any. 
Assuming that for all solutions $\widehat{\varkappa}$ the non-degeneracy condition 
\[
\int_0^t\Omega_{\varkappa\varkappa}(s,\widehat{\varkappa})ds\neq0
\] 
is satisfied and using the standard formulae (see {\it e.g.} \cite[Section 2.9]{Erdelyi}) of the method of stationary phase, 
we infer the following asymptotics as $\varepsilon\to0:$
\begin{equation}
\begin{aligned}
u^\varepsilon(x,t)&\sim\sum_{\widehat{\varkappa}(x,t)}
\widetilde{u}\biggl(x\mp\int_0^t\Omega_\varkappa(s,\widehat{\varkappa})ds,\widehat{\varkappa}\biggr)\sqrt{\frac{\Omega(0,\widehat{\varkappa})}{\Omega(t,\widehat{\varkappa})}}
U^{(0)}\biggl(\frac{x}{\varepsilon},\Omega(t,\widehat{\varkappa})\biggr)
\biggl\vert\int_0^t\Omega_{\varkappa\varkappa}(s,\widehat{\varkappa})ds\biggr\vert^{-1/2}
\\[0.4em]
&
\ \ \ \ \ \ \ \ \ \ \ \ \ \ \ \ \ \ \ \ \ \times\exp\biggl[\frac{\rm i}{\varepsilon}\biggl(\widehat{\varkappa}x\mp\int_0^t\Omega(s,\widehat{\varkappa})ds\biggr)
-\frac{\pi}{4}{\rm sgn}\biggl(\int_0^t\Omega_{\varkappa\varkappa}(s,\widehat{\varkappa})ds
\biggr)\biggr],\\[0.4em]
&=\sum_{\widehat{\varkappa}(x,t)}\widetilde{u}(0,\widehat{\varkappa})\sqrt{\frac{\Omega(0,\widehat{\varkappa})}{\Omega(t,\widehat{\varkappa})}}
U^{(0)}\biggl(\frac{x}{\varepsilon},\Omega(t,\widehat{\varkappa})\biggr)
\biggl\vert\int_0^t\Omega_{\varkappa\varkappa}(s,\widehat{\varkappa})ds\biggr\vert^{-1/2}
\\[0.4em]
&\ \ \ \ \ \ \ \ \ \ \ \ \ \ \ \ \ \ \ \ \ \times\exp\biggl[\frac{\rm i}{\varepsilon}\biggl(\widehat{\varkappa}x\mp\int_0^t\Omega(s,\widehat{\varkappa})ds\biggr)
-\frac{\pi}{4}{\rm sgn}\biggl(\int_0^t\Omega_{\varkappa\varkappa}(s,\widehat{\varkappa})ds
\biggr)\biggr] 
\end{aligned}
\label{wave_formula}
\end{equation}
where ``${\rm sgn}$'' stands for the sign function, and the sum is set to zero if for $x,t$ there are no $\varkappa$ satisfying (\ref{stat_eqn1}).

Suppose, in particular, that 
\begin{equation}
\widetilde{u}(x, \varkappa)=f(x)\delta({\varkappa}-\varkappa_*),
\label{point}
\end{equation}
where $\delta$ is the usual Dirac delta-function, {\it i.e.} the initial data oscillate on the scale $\varepsilon$ as a  wave with quasimomentum $\varkappa_*,$ enveloped by the amplitude function $f=f(x).$ Then the formula (\ref{wave_formula}) yields an asymptotically small (of order $O(\varepsilon)$) value for 
$u^\varepsilon(x,t)$ at all points $(x,t)$ in space-time except ({\it cf.} (\ref{stat_eqn1}))
\[
x=\pm\int_0^t\Omega_\varkappa(s, \varkappa_*)ds,
\]
{\it i.e.} those for which $\widehat{\varkappa}(x,t)=\varkappa_*.$ In the case when $a(t,y)=a(y),$ {\it cf.} (\ref{t_independent}), and assuming that $\Omega'$ is monotonic, this results in the following formula for $u^\varepsilon:$
\begin{equation}
u^\varepsilon(x,t)\sim \frac{f(0)}{\sqrt{t\vert\Omega''(\varkappa_*)\vert}}\delta\bigl(x\mp\Omega'(\varkappa_*)t\bigr)
U^{(0)}\biggl(\frac{x}{\varepsilon},\Omega(\varkappa_*)\biggr)
\exp\biggl[\frac{\rm i}{\varepsilon}\bigl(\varkappa_*x\mp \Omega(\varkappa_*)t\bigr)
-\frac{\pi}{4}{\rm sgn}\bigl(\Omega''(\varkappa_*)\bigr)\biggr], 
\label{point_solution}
\end{equation}
which is a pulse supported at $x=\pm\Omega'(\varkappa_*)t$ ({\it i.e.} moving with velocity $\Omega'(\varkappa_*)$), with amplitude exhibiting two kinds of behaviour in time: decay $1/\sqrt{t\vert\Omega''(\varkappa_*)\vert}$ and oscillation 
$U^{(0)}\bigl(\Omega'(\varkappa_*)t/\varepsilon,\Omega(\varkappa_*)\bigr).$

In the next two sections we show that the property illustrated in (\ref{point_solution}) is general, {\it i.e.} in a wave train the energy locally propagates with the velocity $\Omega'(\varkappa)$ (Section \ref{prop_k}) and that the amplitude modulation of the solution (\ref{point_solution}) leads to a new and potentially useful effect in the case of a high-contrast periodic medium (Section 
\ref{hc_section}).
 
\section{Propagation of local quasimomenta $\varkappa$ and wavenumbers $k$}
\label{prop_k}

Note that the equation (\ref{phi_0_eqn}) can be written as 
\beq
\omega=\mp\Omega(t, k).
\eeq{dispersion_rel}
where we 
denote by 
\begin{align}
\omega&=\omega(x,{t},\varkappa):=-\theta_{t}(x, {t}, \varkappa)=-{\phi}_{t}(x, {t}, \varkappa),
\label{local_defin1}\\[0.3em]
k&=k(x,{t},\varkappa):=-\theta_x(x, {t}, \varkappa)=\varkappa-{\phi}_x(x, {t}, \varkappa)
\label{local_defin2}
\end{align}
the local values of ``frequency'' and ``wavenumber'' in a nonuniform wave train, in particular, in a ``wave packet'' such as (\ref{wave_formula}), {\it cf.} \cite{Whitham}.
We assume that for all $x,{t},\varkappa$ the ``ampltude function'' $u^{(0)}$ does not include any phase of the function 
$\widehat{u}$ by requiring that 
\[
u^{(0)}(x, {t}, \varkappa)=\bigl\vert u^{(0)}(x, {t}, \varkappa)\bigr\vert.
\]
This requirement is met by including the expression for the corresponding phase into the function ${\phi}.$ 

Differentiating the equation (\ref{dispersion_rel}) with respect to $\varkappa$ yields
\beq
\omega_\varkappa=\mp k_\varkappa\Omega_\xi.
\eeq{omega_k}
For the function $\varkappa=\widehat{\varkappa}(x, {t})$ describing the stationary value of $\varkappa$ in (\ref{stat_eqn}), one has, by differentiating  the equation
\begin{equation}
{\phi}_\varkappa(t,x,\widehat{\varkappa})=x
\label{stat_eqn}
\end{equation}
with respect to $x,$
\[
{\phi}_{\varkappa x}+{\phi}_{\varkappa\varkappa}\widehat{\varkappa}_x=1,
\]
or, equivalently, by additionally using (\ref{local_defin2}),
\beq
{\phi}_{\varkappa\varkappa}\widehat{\varkappa}_x=k_\varkappa.
\eeq{comb1}
Further, differentiating (\ref{stat_eqn}) with respect to ${t}$ for fixed $x,$ we write
\[
{\phi}_{\varkappa{t}}+{\phi}_{\varkappa\varkappa}\widehat{\varkappa}_{t}=0,
\]
from which, using the definition (\ref{local_defin1}), we obtain
\beq
\omega_\varkappa-{\phi}_{\varkappa\varkappa}\widehat{\varkappa}_{t}=0.
\eeq{comb2}
Finally, combining (\ref{comb2}), (\ref{omega_k}) and (\ref{comb1}) yields
\[
{\phi}_{\varkappa\varkappa}\widehat{\varkappa}_{t}=\omega_\varkappa=\mp k_\varkappa\Omega_\xi=\mp{\phi}_{\varkappa\varkappa}\widehat{\varkappa}_x\Omega_\xi,
\]
and hence 
\beq
\widehat{\varkappa}_{t}\pm\widehat{\varkappa}_x\Omega_\xi=0,
\eeq{cons_law_varkappa}
assuming that ${\phi}_{\varkappa\varkappa}$ does not vanish in the domain of ${\phi}.$
A version of the transport equation
\beq
k_{t}\pm k_x\Omega_\xi=0,
\eeq{cons_law_k}
for the local wave number described in \cite{Whitham} also holds for the quantity $\widehat{k}$ given by 
({\it cf.} (\ref{local_defin2})) 
\[
\widehat{k}(x,{t}):=k\bigl(x, {t}, \widehat{\varkappa}(x,{t})\bigr)=\widehat{\varkappa}(x,{t})
-{\phi}_x\bigl(x,{t},\widehat{\varkappa}(x,{t})\bigr),
\] 
namely
\[
\widehat{k}_{t}\pm\widehat{k}_x\Omega_\xi=0.
\]
This is obtained immediately by differentiating the equation ({\it cf.} (\ref{dispersion_rel}))
\[
\widehat{\omega}=\mp\Omega(t,\widehat{k}),\ \ \ \ \widehat{\omega}:=\omega\bigl(x, {t}, \widehat{\varkappa}(x,{t})\bigr)
\]
with respect to $x,$ and by noting first that
\beq
\omega_x=-\theta_{{t} x}=k_{t}={\phi}_{x{t}}=\pm{\phi}_{xx}\Omega_\xi
\eeq{later_use}
in view of (\ref{local_defin1}), (\ref{local_defin2}), (\ref{phi_0_eqn}),
and second that 
\[
\omega_{\varkappa}=-{\phi}_{\varkappa{t}}=\mp(1-{\phi}_{x\varkappa})\Omega_\xi
\]
in view of (\ref{phi_0_eqn}).

The equations (\ref{cons_law_varkappa}) and (\ref{cons_law_k}) are interpreted in the sense that the local quasimomentum $\widehat{\varkappa}$ and the local wavenumber $\widehat{k}$ propagate at each point 
$(x,{t})$ with the ``group velocity'' $\Omega_\xi(t, \widehat{k}).$  As we shall see in the next section,
the quantity $\Omega_\xi(t, \widehat{k})$ describes the speed of propagation of the 
wave energy in the wave-train: the amount of the energy carried between two points moving with group velocities remains unchanged with time.

\section{Transport of wave amplitude}
\label{prop_amp}

First we note that the points $x$ that have a fixed value of $\varkappa$ are transported with the group velocity $\Omega_\xi(t, \widehat{k}).$ Indeed, differentiating (\ref{stat_eqn}) with respect to ${t}$ for fixed 
$\varkappa$ we write
\[
{\phi}_{\varkappa x}x_{t}+{\phi}_{\varkappa{t}}=x_{t}.
\]
At the same time, differentiating (\ref{phi_0_eqn}) with respect to $\varkappa$ we obtain
\[
{\phi}_{{t}\varkappa}=\pm(1-{\phi}_{x\varkappa})\Omega_\xi.
\]
Combining the above two equalities yields
\[
(1-{\phi}_{x\varkappa})x_{t}=\pm(1-{\phi}_{x\varkappa})\Omega_\xi,
\]
where $\Omega_\xi(t,\xi)$ is evaluated at $\xi=\varkappa-{\phi}_x\bigl(x({t}, \varkappa), {t}, \varkappa\bigr).$ Hence (assuming that ${\phi}_{x\varkappa}\neq1$) 
\beq
x_{t}({t},\varkappa)=\pm\Omega_\xi,
\eeq{x_speed}
as claimed.

Now, consider the integral of the modulus of the solution $u^\varepsilon$ squared, between any two points $x_1=x_1({t}),$ $x_2=x_2({t})$ moving with 
group velocities corresponding to the values $\varkappa_1$ 
$\varkappa_2$ of the quasimomentum (and hence have the local values of the quasimomentum $\widehat{\varkappa}(x_1,{t})=\varkappa_1,$  
$\widehat{\varkappa}(x_2,{t})=\varkappa_2$ constant in time). Using the asymptotic formula (\ref{wave_formula}), we write, as 
$\varepsilon\to0,$ 
\begin{equation}
\begin{aligned}
Q({t}):=\int_{x_1({t})}^{x_2({t})}&\bigl\vert u^\varepsilon(x)\bigr\vert^2dx
\sim\
\int_{x_1({t})}^{x_2({t})}\Bigl\{u^{(0)}\bigl(x, {t}, \widehat{\varkappa}(x, t)\bigr)\Bigr\}^2\\[0.4em]
&\times\biggl(\int_0^1\Bigl\{U^{(0)}\bigl(y, \widehat{\varkappa}-{\phi}_x\bigl(x,{t},\widehat{\varkappa}(x,t)\bigr)\bigr)\Bigr\}^2dy\biggr)
\Bigl\vert{\phi}_{\varkappa\varkappa}\bigl(x,{t},\widehat{\varkappa}(x,t)\bigr)\Bigr\vert^{-1}dx,
\end{aligned}
\label{energy_exp}
\end{equation}
where we use the result of \cite[Appendix C]{Sm_Cher_2000} to separate averages with respect to the fast and slow variables.
Making the change of the variable from $x$ to $\varkappa$ according to the equation (\ref{stat_eqn}) results in 
\beq
Q({t})\sim
\int_{\varkappa_1}^{\varkappa_2}\Bigl(u^{(0)}\bigl(x({t},\varkappa), {t}, \varkappa\bigr)\Bigr)^2F\bigl(\tilde{k}({t}, \varkappa)\bigr)\Bigl(1-{\phi}_{\varkappa x}\bigl(x({t}, \varkappa), {t},\varkappa\bigr)\Bigr)^{-1}d\varkappa,
\eeq{Q_expression}
where
\[
F(\tilde{k}):=\int_0^1\bigl(U^{(0)}(y, \tilde{k})\bigr)^2dy,\ \ \ \ \ \ \ \ \tilde{k}({t},\varkappa):=\varkappa-{\phi}_x\bigl(x({t},\varkappa),{t},\varkappa\bigr).
\]

\begin{theorem}
The energy (\ref{Q_expression}) is constant in time ${t}.$ 
\end{theorem}
\begin{proof}
Indeed, for the derivative with respect to ${t}$ of the expression under the integral in (\ref{Q_expression}), one has
\begin{align*}
&{\mathcal D}:=\frac{\partial}{\partial{t}}\biggl\{\Bigl(u^{(0)}\bigl(x({t},\varkappa), {t}, \varkappa\bigr)\Bigr)^2F\bigl(\tilde{k}({t}, \varkappa)\bigr)\Bigl(1-{\phi}_{\varkappa x}\bigl(x({t}, \varkappa), {t},\varkappa\bigr)\Bigr)^{-1}\biggr\}\\[0.3em]
&
\qquad\qquad
=\biggl\{\Bigl[\bigl(u^{(0)}\bigr)^2\Bigr]_xx_{t}+\Bigl[\bigl(u^{(0)}\bigr)^2\Bigr]_{t}\biggr\}F(\tilde{k})
\Bigl(1-{\phi}_{\varkappa x}\bigl(x({t}, \varkappa), {t}, \varkappa\bigr)\Bigr)^{-1}\\[0.3em]
&
\qquad\qquad\qquad\qquad
+\bigl(u^{(0)}\bigr)^2F(\tilde{k})\Bigl(1-{\phi}_{\varkappa x}\bigl(x({t}, \varkappa), {t},\varkappa\bigr)\Bigr)^{-2}
\bigl({\phi}_{\varkappa xx}x_{t}+{\phi}_{\varkappa x{t}}\bigr)
\\[0.3em]
&\qquad\qquad\qquad\qquad\qquad\qquad
+\bigl(u^{(0)}\bigr)^2\Bigl(1-{\phi}_{\varkappa x}\bigl(x({t}, \varkappa), {t},\varkappa\bigr)\Bigr)^{-1}
F'(\tilde{k})\bigl(\tilde{k}_xx_{t}+\tilde{k}_{t}\bigr).
\end{align*}
The last term in the above expression vanishes, due to the fact that
\beq
\tilde{k}_xx_{t}+\tilde{k}_{t}=-{\phi}_{xx}x_{t}-{\phi}_{x{t}}=0,
\eeq{ktilde_eq}
which holds by virtue of (\ref{x_speed}) and (\ref{phi_0_eqn}), {\it cf.} (\ref{later_use}).
Further, notice that\footnote{As before, expressions $\Omega_\xi$, $\Omega_{\xi\xi}$ are evaluated at  
$\varkappa-{\phi}_x(x,{t},\varkappa).$}
\[
{\phi}_{\varkappa xx}x_{t}+{\phi}_{\varkappa x{t}}=\mp{\phi}_{xx}(\Omega_\xi)_\varkappa=\mp(1-{\phi}_{x\varkappa}){\phi}_{xx}\Omega_{\xi\xi}=\pm(1-{\phi}_{x\varkappa})(\Omega_\xi)_x,
\]
in view of 
\[
{\phi}_{{t} x\varkappa}=\pm(-{\phi}_{xx})(\Omega_\xi)_\varkappa\pm(-{\phi}_{xx\varkappa})\Omega_\xi,
\]
which, in turn,  is obtained by differentiating the last equality in (\ref{later_use}) with respect to $\varkappa,$ 
{\it cf.} (\ref{ktilde_eq}).

Combining the above observations and the equation (\ref{x_speed}) yields
\[
{\mathcal D}=\biggl\{\Bigl[\pm\bigl(u^{(0)}\bigr)^2\Omega_\xi\Bigr]_x+\Bigl[\bigl(u^{(0)}\bigr)^2\Bigr]_{t}\biggr\}
\Bigl(1-{\phi}_{\varkappa x}\bigl(x({t}, \varkappa), {t}, \varkappa\bigr)\Bigr)^{-1},
\]
which vanishes thanks to the transport equation (\ref{u0_final}) for the function $u^{(0)}.$ This concludes the proof.
\end{proof}


The above argument implies, in particular, that 
\begin{equation}
{\mathcal E}_{t}+{\mathcal F}_x=0,
\label{energy_flux}
\end{equation}
where ${\mathcal E}$ is the energy density, given for each $(x, {t})$ by the expression under the integral 
in (\ref{energy_exp}), and 
\[
{\mathcal F}=\pm\Omega_\xi\bigl(\widehat{k}(x,{t})\bigr){\mathcal E},\qquad\widehat{k}(x, {t}):=\widehat{\varkappa}(x, {t})-
{\phi}\bigl(x,{t},\widehat{\varkappa}(x,{t})\bigr), 
\]
is the density of the ``energy flux''.The formula (\ref{energy_flux}) shows that the energy is carried by the wave packet  with the group velocity corresponding to the local value $\widehat{\varkappa}$ of the quasimomentum.

\section{Wave modulation in a high-contrast periodic medium}
\label{hc_section}

In this section we discuss a class of piecewise-constant coefficients $a$ in (\ref{1D_orig_eqn}) 
\begin{equation}
a(y,t)=\left\{\begin{array}{ll}a_1,\quad y\in(0, h),\\[0.25em]a_2,\quad y\in(h, 1),\end{array}\right.
\label{coef_a}
\end{equation}
where one of the the values, say $a_1,$ is assumed to be large. In the applied analysis literature this kind of model is sometimes referred to as the ``large-coupling" limit of the problems (\ref{1D_orig_eqn}), see \cite{HL, Cher_Kis_Silva}. In the context of homogenisation, when $a_1=\varepsilon^{-2}$ this also corresponds to the ``critical high contrast" limit, when the medium exhibits ``metamaterial'' behaviour, see \cite{Physics, CEK, SK} for details.   The spectral and resolvent analysis of the one-dimensional setup has been carried out in \cite{CherednichenkoCooperGuenneau, CherKis, GrandePreuve, CCC}.

The limit ($a_1\to\infty$) dispersion relation $\Omega=\Omega(\varkappa)$ for (\ref{coef_a}) is given implicitly by (see {\it e.g.} \cite{CherednichenkoCooperGuenneau})
\begin{equation}
\cos\biggl(\dfrac{\Omega}{\sqrt{a_2}}(1-h)\biggr)-\frac{1}{2}\sin\biggl(\dfrac{\Omega}{\sqrt{a_2}}(1-h)\biggr)\dfrac{\Omega h}{\sqrt{a_2}}=\cos\varkappa,\qquad \varkappa\in[-\pi, \pi).
\label{hc_dispersion}
\end{equation}
and the corresponding eigenfunction limit is shown to be given by
\begin{equation}
C_{\varkappa,h}U^{(0)}(y, \Omega)=\left\{\begin{array}{ll}\sin\biggl(\dfrac{\Omega}{\sqrt{a_2}}(1-h)\biggr)\exp(-{\rm i}\varkappa y),\qquad y\in(0,h],\\[1.0em]
\sin\biggl(\dfrac{\Omega}{\sqrt{a_2}}(1-y)\biggr)+\sin\biggl(\dfrac{\Omega}{\sqrt{a_2}}(y-h)\biggr)\exp\bigl({\rm i}\varkappa(1-y)\bigr),\qquad y\in(h, 1),\end{array}\right.
\label{eigenf}
\end{equation}
where $C_\varkappa$ is a normalising coefficient, which ensures that (\ref{eigen0_norm}) holds, see (\ref{ckappa}) below.
The set of values $\Omega$ satisfying (\ref{hc_dispersion}) for a given $\varkappa$ is a sequence of values 
$\{\Omega_n(\varkappa)\}_{n\in{\mathbb N}}$, which sweeps a countable union of disjoint intervals separated by gaps, see Figure \ref{limitspectrum}. 
\begin{center}
\begin{figure}[h]
\centering
\includegraphics[scale=0.4]{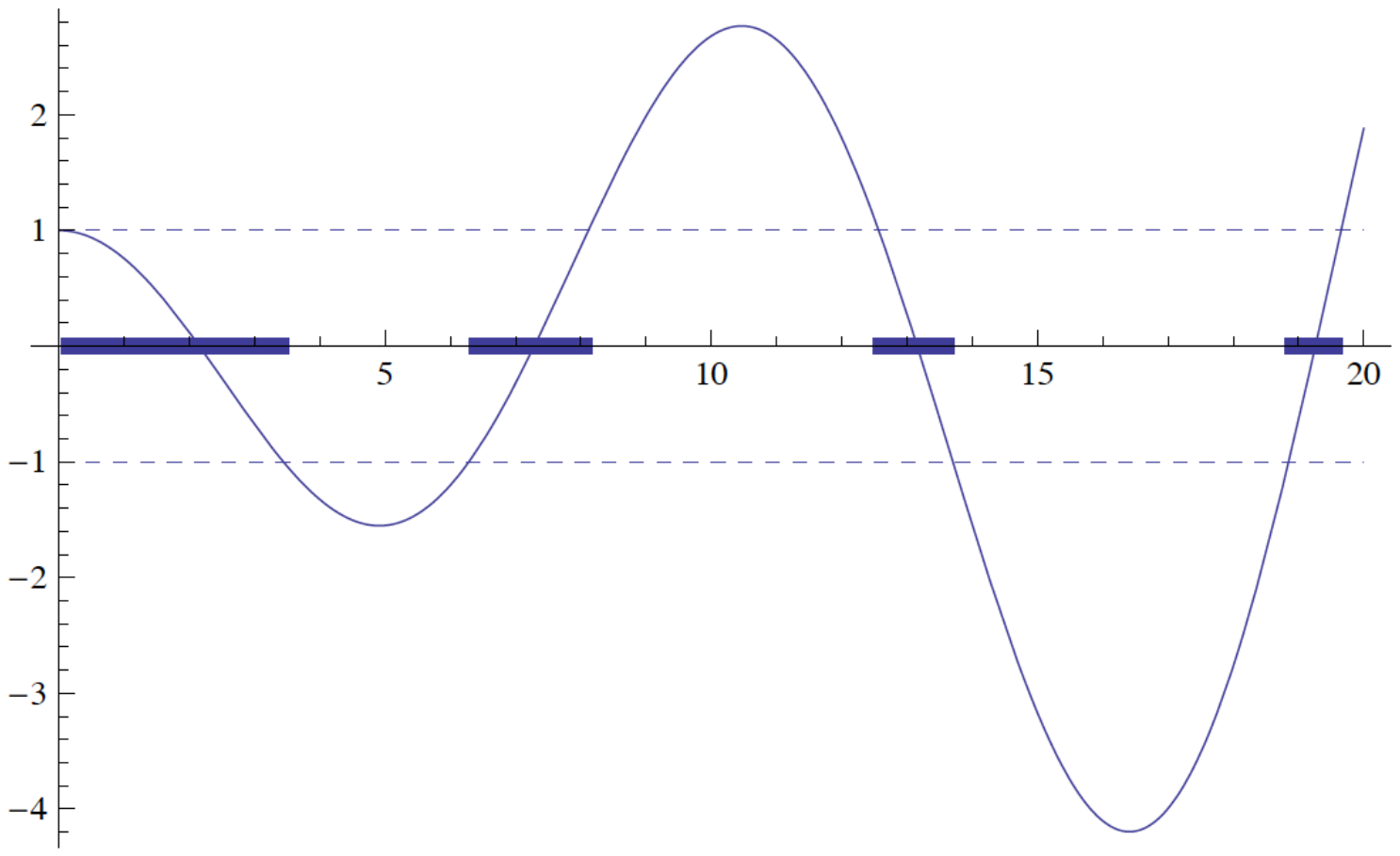}
\caption{The square root of the limit spectrum for a high-contrast periodic medium, {\it i.e.} the set of $\Omega$ that satisfy (\ref{hc_dispersion}) for some $\varkappa\in[-\pi, \pi).$ The oscillating solid line is the graph of the function
$f(\Omega)=\cos(\Omega/2)-\Omega\sin(\Omega/2)/4,$ which corresponds to setting $h=1/2,$ $a_2=1$ 
in the formula (\ref{hc_dispersion}). 
The square root of the spectrum is the union of the intervals indicated by bold lines, consisting of values $\Omega\in{\mathbb R}^+$ such that $|f(\Omega)|\le 1.$}
\label{limitspectrum} 
\end{figure}
\end{center}

The spectral intervals get narrower as $n\to\infty,$ while the dispersion curves $\Omega=\Omega_n(\varkappa)$ get flatter: it is straightforward to see that, labelling by $n$ the $(n+1)$st band, $n\in{\mathbb N},$ one has
\[
\Omega_n(\varkappa)=\frac{\sqrt{a_2}n\pi}{1-h}+\frac{2\sqrt{a_2}}{n\pi h }\bigl(1+(-1)^{n+1}\cos\varkappa\bigr)+O\biggl(\frac{1}{n^2}\biggr),\qquad n\to\infty.
\] 
Similarly, by differentiating (\ref{hc_dispersion}), it is shown that 
\[
\Omega_n'(\varkappa)=(-1)^n\frac{2\sqrt{a_2}}{n\pi h}\sin\varkappa+O\biggl(\frac{1}{n^2}\biggr),\qquad \Omega_n''(\varkappa)=(-1)^n\frac{2\sqrt{a_2}}{n\pi h}\cos\varkappa+O\biggl(\frac{1}{n^2}\biggr)\qquad n\to\infty.
\]

Furthermore, for the two eigenfunctions (\ref{eigenf}) corresponding to the eigenvalue $\Omega$ we obtain
\[
\sin\biggl(\dfrac{\Omega}{\sqrt{a_2}}(1-h)\biggr)=\frac{2(1-h)}{n\pi h}(-1)^{n+1}\bigl(1+(-1)^{n+1}\cos\varkappa\bigr)+O\biggl(\frac{1}{n^2}\biggr),\qquad n\to\infty,
\]
and  
\begin{equation}
C_{\varkappa, h}=\sqrt{(1-h)\bigl(1+(-1)^{n+1}\cos\varkappa\bigr)}+O\biggl(\frac{1}{n^2}\biggr),\qquad n\to\infty,
\label{ckappa}
\end{equation}

Consider the setup discussed at the end of Section \ref{asymptotic_form}: a slowly modulated $\varepsilon$-oscillatory wave described by an eigenfunction corresponding to a specified value of the quasimomentum $\varkappa_*,$ see (\ref{init_spec}), (\ref{point}). Using the formula (\ref{point_solution}) we infer that for $x\in \varepsilon(l, l+h),$ $l\in{\mathbb Z},$ {\it i.e.} on the ``stiff'' intervals, one has, as $\varepsilon\to0,$
\begin{equation}
\begin{aligned}
u^\varepsilon(x,t)\sim
(-1)^{n+1}\frac{f(0)}{a_2^{1/4}}&\sqrt{\frac{2(1-h)\bigl(1+(-1)^{n+1}\cos\varkappa_*\bigr)}{tn\pi h|\cos\varkappa_*|}}\\[0.6em]
&\times\exp\biggl[{\rm i}\biggl(\varkappa_*l\mp\dfrac{\Omega_n(\varkappa_*)t}{\varepsilon}\biggr)
-\dfrac{\pi}{4}{\rm sgn}\bigl(\Omega_n''(\varkappa_*)
\bigr)\biggr]\delta\bigl(x-\Omega_n'(\varkappa_*)t\bigr),
\end{aligned}
\label{form_stiff}
\end{equation}
which is a pulse of constant amplitude, proportional to $\sqrt{(1-h)/(tn\pi h)}.$
At the same time, for $x\in \varepsilon(l+h, l+1),$ $l\in{\mathbb Z},$ {\it i.e.} on the ``soft" intervals, one has, as $\varepsilon\to0,$
\begin{equation}
\begin{aligned}
u^\varepsilon(x,t)\sim
\frac{f(0)}{a_2^{1/4}}
&\sqrt{\frac{n\pi h}{2t|\cos\varkappa_*|\bigl(1+(-1)^{n+1}\cos\varkappa_*\bigr)}}
\\[0.6em]
&\times\biggl\{\sin\biggl(\dfrac{\Omega_n(\varkappa_*)}{\sqrt{a_2}}(1-x)\biggr)+\sin\biggl(\dfrac{\Omega_n(\varkappa_*)}{\sqrt{a_2}}(x-h)\biggr)\exp({\rm i}\varkappa_*)\biggr\}
\\[0.6em]
&\times\exp\biggl[{\rm i}\biggl(\varkappa_*l\mp\dfrac{\Omega_n(\varkappa_*)t}{\varepsilon}\biggr)
-\dfrac{\pi}{4}{\rm sgn}\bigl(\Omega_n''(\varkappa_*)
\bigr)\biggr]\delta\bigl(x-\Omega_n'(\varkappa_*)t\bigr),
\end{aligned}
\label{form_soft}
\end{equation}
{\it i.e.} a pulse with an oscillatory amplitude with maxima proportional to $\sqrt{n\pi h/((1-h)t).}$ In particular, for a fixed $\varkappa_*,$ the ratio of the maximal values of the pulse amplitude in the stiff and soft components is $2(1-h)/{n\pi h}.$ 

The formulae (\ref{form_stiff})--(\ref{form_soft}) show that by 
choosing
$n$ to be large large ({\it i.e.} high frequency of the transmission band), the amplitude of the travelling pulse can be reduced in the stiff intervals and amplified in the soft intervals. Furthermore, by choosing $n$ to be odd and tuning $\varkappa_*$ to be close to $\pi,$ or by choosing $n$ to be even and $\varkappa^*$ to be close to zero, a ``resonance'' occurs, where the maximum of the pulse amplitude in the soft component blows up to infinity while vanishing in the stiff component. 


\section*{Acknowledgements} The author is grateful for the financial support of the Engineering and Physical Sciences Research Council: Grant EP/L018802/2 ``Mathematical foundations of metamaterials: homogenisation, dissipation and operator theory''.
He is also grateful to the Isaac Newton Institute for Mathematical Sciences, Cambridge, for support and hospitality during the programme ``Periodic and Ergodic Spectral Problems'', where work on this paper was partially undertaken, to Professor Graeme Milton for suggesting the problem studied in this article, and to the Department of Mathematics, University of Utah, for hospitality. 

\end{document}